\documentclass[11pt, a4paper]{article}
\usepackage{my_style}
\usepackage{aidan_proposed_macros}
\usepackage{tablefootnote} %

\newcommand{\GF}[1]{\mathrm{GF}(#1)}
\newcommand{\Ftwo}{\mathbb{F}_2}
\newtheorem*{mainthm}{Main Theorem}

\title{Encoded reversible addition in constant Toffoli depth}  
\author[1]{Aidan Epperly}
\author[2]{Isaac H. Kim}
\author[3]{Mingyu Lee}
\affil[1]{Department of Mathematics, University of California, Davis, CA, 95616, USA}
\affil[2]{Department of Computer Science, University of California, Davis, CA, 95616, USA}
\affil[3]{Department of Computer Science and Engineering, Seoul National University, Seoul 08826, Republic of Korea}
\date{}
\begin{document}
\maketitle
\begin{abstract}
We show that reversible in-place quantum addition modulo $2^n-1$ can be implemented in constant Toffoli depth using $O(n)$ space, provided that the variables are encoded in a certain non-standard form. Using measurement-based uncomputation, the Toffoli depth can be reduced to one. We also present space--depth tradeoffs and an application to programmable rotations.
\end{abstract}

\section{Introduction}
\label{sec:intro}

Integer addition is one of the most commonly used computational primitives. In the standard binary representation, it is a well-known fact that the textbook method achieves this in depth linear in the number of bits. In fact, it is possible to reduce this depth to be logarithmic in the number of bits, which is optimal for circuits of bounded fan-in~\cite{winograd1965time}. The lower bound seemingly closes this problem completely, but more can be said if we do not demand the number to be stored in the standard binary representation. Indeed, von Neumann's carry save adder can achieve addition in constant depth~\cite{BurksGoldstineVonNeumann1946}. There are alternatives, such as signed-digit representation, that achieve a similar goal~\cite{avizienis1961signed}.

In the context of quantum computing, integer addition continues to play an important role in optimization of quantum algorithms~\cite{Cuccaro2004,Draper2006,Gidney2018}. In the standard binary representation, quantum adders with $O(\log n)$ depth are known, both modulo $2^n$ and modulo $2^n-1$~\cite{Draper2006,Gaur2024}.\footnote{If unbounded fan-out gates are allowed, binary addition can even be done in constant depth~\cite{hoyer2005quantum,takahashi2016collapse}, but these constructions use single-qubit rotations by angles that shrink with $n$, which are not part of the Clifford+Toffoli gate set.} As such, it is natural to ask if the addition circuits developed in the classical world can be used to optimize quantum algorithms as well. To that end, it is important to note that in quantum circuits classical computations must be performed reversibly. This implies, for instance, that the carry save adder or the signed-digit representation, is inadequate. Although these representations allow addition to be done in constant depth, the same integer can have several distinct encodings. 

In order to explain this issue, it will be helpful to formulate the problem as follows. Consider integer addition modulo $N\in \mathbb{N}$. An encoding assigns a bit string $\mathcal{E}(x)$ to each $x\in\mathbb{Z}_N$, and an encoded addition maps the encodings of $x$ and $y$ to an encoding of $x+y$. If every integer has exactly one encoding, we say the encoded addition is \emph{reversible}. (See Section~\ref{sec:addition} for the formal definition.)

In this language, the carry-save adder and signed-digit representation form an encoded addition scheme which is \emph{not} reversible. This is for the simple reason that there can be two different encodings of the same integer. If we were to use such an encoding in quantum computation, we may run into some issues. While the user may have intended to induce constructive interference, it may happen that the encodings of the two branches are different. If so, the two branches remain orthogonal, which can potentially cause issues downstream. As such, in quantum computation, one must use an encoding scheme that is necessarily reversible. One such encoding is the residue number system~\cite{Garner1959}, which stores an integer by its residues modulo several coprime moduli, so that each residue can be added independently. However, the largest modulus must then have $\Omega(\log n)$ bits, so with known adders this gives Toffoli depth $O(\log\log n)$ rather than a constant.

With these considerations in mind, the main question at hand is the ultimate cost of encoded reversible addition. A common way of quantifying this cost is in terms of the standard elementary gate set in quantum computing, which is Clifford+$T$, or alternatively, Clifford+Toffoli. The $T$-gate and Toffoli are examples of the non-Clifford gates, which are essential in completing the universal gate set. Such operations have been considered significantly more expensive than Clifford gates~\cite{BravyiKitaev2005}. For this reason, the standard cost metric has been the number of non-Clifford gates used in the computation. However, the cost of non-Clifford gates has recently been reduced substantially~\cite{litinski2019magic}, motivating studies that minimize the depth of the non-Clifford gates~\cite{goto2016minimizing,Chamberland2020,Itogawa2024,gidney2024magicstatecultivationgrowing,Daguerre2025,daguerre2025experimentaldemonstrationhighfidelitylogical,dasu2025breakingmagicdemonstrationhighfidelity}. 

The main result of our paper is that there is a reversible encoded addition scheme modulo $N=2^n-1$ that admits a constant Toffoli depth, using $O(n)$ ancillas.
\begin{mainthm}[informal]
For every $n\ge 2$, there is an injective encoding of integers modulo
$2^n-1$ using $O(n)$ qubits under which exact in-place addition has $O(1)$ Toffoli depth, uses $O(n)$ total qubits, and requires $O(n)$
Toffoli gates. Moreover, the circuit can be constructed in time polynomial in $n$.
\end{mainthm}
\noindent
See Theorems~\ref{thm:depth_four_inplace} and~\ref{thm:main_inplace}. For the optimized constants, we refer to Section~\ref{sec:optimized}.

Our result is enabled by a somewhat counterintuitive approach of replacing integer addition by multiplications in finite fields. While multiplication is typically more expensive than addition, for finite fields there is an interesting algorithm that can substantially reduce the multiplication cost~\cite{ChudnovskyChudnovsky1988}. We will leverage this fact throughout this paper. We will also discuss various strategies for optimizing the constants of our scheme. 

Quantum circuits for multiplication in $\GF{2^n}$ have been studied extensively, mainly in the context of Shor's algorithm~\cite{shor1994algorithms} for elliptic curves over binary fields~\cite{Maslov2009,RoettelerSteinwandt2014,KepleySteinwandt2015,vanHoof2020,Banegas2021,Vandaele2025}. In particular, it has been observed that the Toffoli count of such circuits is governed by the bilinear complexity, and multipliers with Toffoli depth one have been constructed~\cite{Jang2023,Liao2025}. Our contribution is to use such multipliers for integer addition, by encoding the integers appropriately, and to make the resulting addition reversible and in place.

The rest of the paper is structured as follows. In Section~\ref{sec:addition} we set up our notation and describe how integer addition modulo $2^n-1$ can be realized by a multiplication over the binary extension field. In Section~\ref{sec:algebraic_geometry} we review results in algebraic geometry that yield a low-complexity multiplication algorithm in this field. In Section~\ref{sec:bilinear_to_circuit} we describe our constant-depth quantum circuit for addition. In Section~\ref{sec:optimized}, we describe methods for optimizing the space-time cost of the addition. In Section~\ref{sec:applications} we discuss applications.

\section{Integer addition from binary extension fields}
\label{sec:addition}
In this section, we formalize the encoded arithmetic tasks considered throughout the paper and then introduce the particular encodings used in our constructions.

Let $M=2^n-1$, and let $\mathbb{Z}_M$ denote the integers modulo $M$.
An encoding consists of a set of valid encodings $\mathcal{C}$ together with a bijection
\begin{equation}
    \mathcal{E}:\mathbb{Z}_M\to\mathcal{C}.
\end{equation}
We write $\mathcal{D}:\mathcal{C}\to\mathbb{Z}_M$ for its inverse on the set of valid encodings, so that
\begin{equation}
    \mathcal{D}(\mathcal{E}(x))=x
\end{equation}
for every $x\in\mathbb{Z}_M$.
Since $\mathcal{E}$ is a bijection, every integer has exactly one valid encoding; this is the sense in which our encoded additions are reversible.
When $\mathcal{E}(x)$ is represented as a binary string, we write
$|\mathcal{E}(x)\rangle$ for the corresponding computational-basis state.
All additions of logical values are understood modulo $M$.

Before we proceed with details, let us first introduce different notions of additions.
First, let us differentiate between in-place and out-of-place addition.
In our language, they can be formulated as follows:
\begin{equation}
\begin{aligned}
    \text{In-place: }\quad
    &|\mathcal{E}(x)\rangle|\mathcal{E}(y)\rangle
    \longmapsto
    |\mathcal{E}(x)\rangle|\mathcal{E}(x+y)\rangle,\\
    \text{Out-of-place: }\quad
    &|\mathcal{E}(x)\rangle|\mathcal{E}(y)\rangle|z\rangle
    \longmapsto
    |\mathcal{E}(x)\rangle|\mathcal{E}(y)\rangle|z\oplus\mathcal{E}(x+y)\rangle.
\end{aligned}
\end{equation}
Here $\oplus$ denotes bitwise XOR, while $x+y$ is evaluated modulo $M$.

We call the qubits in the encoded input/output registers \emph{data qubits},
and any additional temporary qubits \emph{workspace qubits}.
Unless stated otherwise, workspace qubits start and end in $|0\rangle$.
Our primary space measure is the total number of qubits, including both
data and workspace qubits.
Throughout, we assume that the numbers to be added are already stored in encoded form and remain so across repeated additions. We therefore do not include the cost of converting to or from the standard binary representation, which amounts to exponentiation and discrete logarithm in $\GF{2^n}$, respectively. In the application of Section~\ref{sec:applications}, the encoded numbers are loaded directly from classical tables, so no such conversion is needed.

Now let us describe our construction, which is based on the binary extension field $\GF{2^n}$. This is a set of polynomials with binary coefficients modulo an irreducible polynomial of degree $n$. The multiplicative group $\GF{2^n}^\times$ of this field is cyclic of order $2^n-1$. Let $f\in \GF{2^n}^\times$ be a generator (a primitive element), which we choose once and for all. Any nonzero element $g\in \GF{2^n}$ can be uniquely represented as $g = f^{\alpha}$ for some $\alpha \in \mathbb{Z}_M$. Here $\alpha$ is the discrete logarithm of $g$, which we denote as $\log g$.

Let $\alpha, \beta \in \mathbb{Z}_M$. We first introduce the following encoding and decoding scheme:
\begin{equation}
    \mathcal{E}_0(\alpha) = f^{\alpha}, \quad \mathcal{D}_0(g) = \log g.
\end{equation}
In other words, the encoding is an exponentiation of the generator polynomial $f$, addition is performed by multiplication in $\GF{2^n}$, and the decoding is the logarithm of the nonzero element. This scheme is reversible because the encoding map is bijective. The identity
\begin{equation}
    \mathcal{E}_0(\alpha)\mathcal{E}_0(\beta)
    = f^\alpha f^\beta
    = f^{\alpha+\beta}
    = \mathcal{E}_0(\alpha+\beta)
\end{equation}
shows that multiplication in $\GF{2^n}$ implements addition modulo $M$ under this encoding.

Thus, under this encoding, modular addition is reduced to multiplication in the binary extension field. The complexity of this operation can be understood in terms of the \emph{bilinear complexity}. Let $u,v \in \GF{2^n}$ and consider the multiplication map $m_n : \GF{2^n} \times \GF{2^n} \to \GF{2^n}$ which is a bilinear map over $\GF{2}$.
From the universal property of the tensor product, there is a unique tensor 
$$
M_n \in \GF{2^n} \otimes \GF{2^n}^* \otimes \GF{2^n}^*
$$
such that $m_n(u,v) = M_n(u \otimes v)$, where here we use $M_n(u \otimes v)$ to denote tensor contraction.
The tensor $M_n$ will henceforth be called the multiplication tensor.

Every tensor can be decomposed nonuniquely as a sum of simple tensors.
Doing so to the multiplication tensor yields a sum 
\begin{align}
    M_n = \sum_{i = 1}^{\mu(n)} c_i \otimes a_i^* \otimes b_i^*
\end{align}
where $\mu(n)$ would ideally be the tensor rank of $M_n$, though in practice it may be larger than the true tensor rank.
Since contraction of tensors is linear, when we contract $M_n(u \otimes v)$ we obtain the form
\begin{align}\label{eq:bilinear_decomposition}
    M_n(u \otimes v) = \sum_{i = 1}^{\mu(n)} c_i (a_i \cdot u)(b_i \cdot v) = C((Au) \circ (Bv))
\end{align}
where $\circ$ denotes the entrywise product,
$A,B\in\GF{2}^{\mu(n)\times n}$, and
$C\in\GF{2}^{n\times\mu(n)}$ has columns $c_i$.
In this way, the matrices $A$, $B$, and $C$ specify an algorithm for multiplying elements of $\GF{2^n}$.
In this algorithm, the matrix-vector products can be computed using CNOT gates, since $A$, $B$, and $C$ are fixed.
However, each of the $\mu(n)$ bilinear multiplications necessary to compute the entrywise product requires a Toffoli gate.
Since $\mu(n)$ is exactly the number of bilinear operations in our algorithm, we will refer to $\mu(n)$ as the \emph{length} of the decomposition; the minimum length over all decompositions is the \emph{bilinear complexity} of multiplication in $\GF{2^n}$.\footnote{We note that for more general boolean functions one can consider a notion of multiplication complexity~\cite{meuli2019role}.}

We note that Eq.~\eqref{eq:bilinear_decomposition} already ensures that the multiplication can be realized in $O(1)$ Toffoli-depth. Indeed, each multiplication $(a_i \cdot u )(b_i \cdot v)$ can be done in parallel for every $i$ after applying the CNOT gates judiciously. Less obvious is whether $\mu(n)$ can be taken to be something that grows more slowly than $n^2$. (The obvious decomposition of the multiplication trivially yields $\mu(n)\le n^2$.) Remarkably, $\mu(n)=O(n)$ is achievable~\cite{ChudnovskyChudnovsky1988,STV1992,Cascudo2012,Ballet2021survey,BBP2022}. This is a known result whose proof uses algebraic geometry, which we review in Section~\ref{sec:algebraic_geometry}.\footnote{For the application-oriented readers, we note that the sections appearing after Section~\ref{sec:algebraic_geometry} are self-contained and do not depend on that discussion.} 

\section{Algebraic Geometry}
\label{sec:algebraic_geometry}

It remains to explain why the bilinear complexity of multiplication in $\GF{2^n}$ is $O(n)$.
In Section~\ref{sec:chudnovsky}, we focus on the practical mathematical background that one would need to understand the algorithm in an expository manner.
In Section~\ref{sec:projective-line-et-al}, we will present constructions of multiplication algorithms with bilinear complexity $O(n)$, why they are not necessarily practical, and discuss a practical construction attaining quasi-linear bilinear complexity.

\subsection{Chudnovsky-type algorithms}
\label{sec:chudnovsky}

Throughout this section, we denote by $[q]_I$ the residue of a ring element $q \in R$ modulo an ideal $I \subseteq R$. In the case that $I = \langle p \rangle$ is principal, we employ the shorthand $[q]_p$ to denote $[q]_I$.

We are primarily interested in studying the multiplication of elements in the Galois field $\GF{2^n}$. 
The most common construction of $\GF{2^n}$ is as the quotient ring $\GF{2}[x] / \langle p\rangle$ of the polynomial ring $\GF{2}[x]$ by the principal ideal $\langle p \rangle$ for $p \in \GF{2}[x]$ an irreducible degree $n$ polynomial. 
Every equivalence class in $\GF{2}[x]/\langle p \rangle$ has a unique representative element of degree at most $n-1$.

Thus we can uniquely identify every element $u \in \GF{2^n}$ with a polynomial $\tilde{u} \in \GF{2}[x]$ of degree at most $n-1$. In this way, multiplication of elements $u,w \in \GF{2^n}$
can be computed by multiplying their polynomial representatives $\tilde{u},\tilde{w} \in \GF{2}[x]$.

Most efficient algorithms for polynomial multiplication are interpolation-based: relying on the fact that the product of two degree $n-1$ polynomials $q,r$ can be determined by the values that $qr$ takes at $2n - 1$ points.
Here we run into a fundamental limitation inherent to finite fields.
There are only $2$ points in $\GF{2}$.
We will work to rectify this by moving from evaluating polynomials at points to evaluating polynomials at ``places''.

Let $F / \GF{2}$ be an algebraic function field: a finite extension of the field of rational functions $\GF{2}(x)$.
A strict subring $\cal O \subset F$ is called a \emph{valuation ring} when $\cal O$ contains the base field $\GF{2}$ and when $f \notin  \cal O$ implies $f\invs \in \cal O$.
Given a valuation ring $\cal O$, the \emph{place} $P$ associated to $\cal O$ is the set of non-units of $\cal O$, which is $P = \cal O \setminus \cal O^{\times}$.
When we refer to a place $P$, we mean the place $P$ associated to some valuation ring $\cal O$.
Any place $P$ is only associated to exactly one valuation ring $\cal O$, since we can recover $\cal O$ from $P$ by the equation 
$$
\cal O = \cal O_P := \set{f \in F \; : \; fP \subseteq P}.
$$
The place $P$ is the unique maximal ideal of $\cal O_P$, and thus the quotient $\cal O_P / P$ is a field.
We denote the residue field $\cal O_P / P$ by $F_P$ and call it the \emph{residue field at} $P$.
This residue field $F_P$ is a finite degree field extension of the base field $\GF{2}$ and thus $F_P \cong \GF{2^d}$ for some $d \geq 1$.
We say that a place $P$ has \emph{degree} $d$ if $F_P \cong \GF{2^d}$.

We can evaluate elements $f \in F$ at the places $P$ of $F$ in the following way.
Fix a place $P$ and the valuation ring $\cal O_P$ associated to it.
If $f \in \cal O_P$, we say that
\begin{align}
    f(P) = [f]_{P} \in F_P,
\end{align}
which is to say that $f(P)$ is the residue of $f$ mod $P$.
Evaluating a function $f$ at $P$ returns a field element, just like evaluation at a point, but the field $F_P$ will generally not be isomorphic to the base field $\GF{2}$.
If the element $f \notin \cal O_P$, then the inverse $f\invs \in \cal O_P$ is a non-unit of $\cal O_P$ and thus is contained in $P$.
By the definition above, then, $f\invs$ takes the value $f\invs(P) = 0$.
In this case, we say that $f$ has a \emph{pole} at the place $P$.

Places faithfully generalize points.
To see this, consider the field of rational functions $\GF{2}(x)$.
Given a rational function $f = q/r \in \GF{2}(x)$ for $q,r \in \GF{2}[x]$ relatively prime and $\alpha \in \GF{2}$, we would typically define $f(\alpha)$ using the equation
\begin{align}
    f(\alpha) = \frac{q(\alpha)}{r(\alpha)} = \frac{\sum_{i} q_i \alpha^i}{\sum_{i} r_i \alpha ^i}
\end{align}
when $r(\alpha) \neq 0$ or declare that $f$ has a pole at $\alpha$ when $r(\alpha) = 0$.
Alternatively, one can consider the local ring $\cal O_\alpha = \set{f = q/r \in \GF{2}(x) \; : \; r(\alpha) \neq 0}$.
This ring is a valuation ring and the ideal $P_\alpha = \langle x-\alpha \rangle \subseteq \cal O_{\alpha}$ is the place associated to it.
If we evaluate $f$ at $P_\alpha$, the value $f(P_\alpha) \in \cal O_\alpha / \langle x - \alpha \rangle \cong \GF{2}$ is then, up to isomorphism, equal to $f(\alpha)$ as defined above.
This can be seen most easily in the case of a polynomial $q$.
We can obtain a residue class representative of $[q]_{x-\alpha}$ by running polynomial long division to obtain $q = s\cdot(x-\alpha) + q(\alpha)$ and thus $[q]_{x-\alpha} = [q(\alpha)]_{x-\alpha}$.
This extends to rational functions $f = q / r$ by $f(P_\alpha) = q(P_\alpha) r(P_\alpha)\invs$.
More generally, given any degree $d$ irreducible polynomial $p \in \GF{2}[x]$, the local ring $\cal O_p = \set{q/r \; : \; p \nmid r}$ is a valuation ring with degree $d$ place $\langle p\rangle \subseteq \cal O_p$.

The basic idea behind the Chudnovsky algorithm over $\GF{2^n}$ is as follows.
Firstly, we need a function field $F / \GF{2}$ that has a degree $n$ place $Q$ so that $F_Q \cong \GF{2^n}$.
This allows us to represent elements $u,w \in \GF{2^n}$ as elements $[\tilde{u}]_Q,[\tilde{w}]_Q \in F_Q \cong \GF{2^n}$ for $\tilde{u},\tilde{w} \in \cal O_Q \subseteq F$.
So, just as in the case of $\GF{2}[x]$, we can multiply two elements of $\GF{2^n}$ by multiplying $\tilde{u},\tilde{w}$ in $F$.
Secondly, the function field that we pick must have enough low-degree places to determine the product $\tilde{u}\tilde{w}$: this is the step analogous to interpolation at points.
The original algorithm of \cite{ChudnovskyChudnovsky1988} required the function field to have a sufficient number of degree one places.
However, more modern versions of the algorithm can make use of places of higher degree \cite{BBP2020projective,Ballet2021survey,BBP2022}.
Finally, we multiply the two elements $\tilde{u}$ and $\tilde{w}$ by multiplying the values they take at each place and then reconstructing the product $\tilde{u}\tilde{w}$, and more importantly $[\tilde{u} \tilde{w}]_Q$, from this local information.

\subsubsection{Valuations, Divisors, and the Riemann--Roch Space}

The description above is a bare-bones summary of the Chudnovsky algorithm.
For the interested reader, there are a number of technical assumptions that are important to the construction.
Of primary importance is that we are working over a function field, not over the ring of polynomials.
Chudnovsky-type multiplication algorithms depend on the existence of a degree $n$ place $Q$ and a collection of low-degree places that we will denote by $\cal P$.
A consequence of working in a function field is that an element $f \in F$ may have a pole at any or all of the places $\cal P \cup \set{Q}$.
We will handle this in essentially the most straightforward way possible: by disallowing elements that have poles at any of the important places.

Every place $P \subseteq F$ has a map $v_P : F \to \mathbb{Z} \cup \set{\infty}$ associated to it called the \emph{valuation map} which satisfies the properties
\begin{enumerate}
    \item $v_P(fg) = v_P(f) + v_P(g)$,
    \item $v_P(f + g) \geq \min \set{v_P(f), v_P(g)}$,
    \item $v_P(f) = \infty \iff f = 0$.
\end{enumerate}
This map is made unique by the further assumptions that $v_P(f) > 0$ implies $f \in P$ and $v_P(t_P) = 1$ for some $t_P \in P$.
The value $v_P(f)$ is called the \emph{order of vanishing} of $f$ at $P$ and any $t_P \in P$ with $v_P(t_P) = 1$ is called a \emph{local parameter} of $P$.
If the value is positive, then we say $f$ has a zero of order $v_P(f)$ at $P$ and if the value is negative, then we say $f$ has a pole of order $-v_P(f)$ at $P$.
We can track the order of $f \in F$ at every place of $F$ using a \emph{divisor}.
A divisor is a formal $\mathbb{Z}$-linear combination of places
\begin{align}
    D = \sum_P n_P P
\end{align}
with only finitely many non-zero terms.
The divisor $D(f)$ of $f \in F^\times$ is the divisor
\begin{align}
    D(f) = \sum_{P} v_P(f) P.
\end{align}
That $D(f)$ is a divisor amounts to a generalization of the fact that a rational function has finitely many zeros and poles.
A divisor $D = \sum n_P P$ is said to be \emph{effective} if $n_P \geq 0$ for all $P$ and we denote this by $D \geq 0$.
We can extend this to a partial ordering on divisors where $D \geq C$ if $D - C \geq 0$.
For a given divisor $D$, the \emph{Riemann--Roch space} $\cal L(D)$ is the set of all elements $f \in F$ such that $D(f) \geq -D$.
We also include $0 \in \cal L(D)$ despite the fact that $D(0)$ is not well-defined (it takes value $\infty$ at every place).
From the properties of a valuation, it can be shown that $\cal L(D)$ is a finite dimensional vector space over $\GF{2}$.
We call this the Riemann--Roch space since its dimension can be computed using the Riemann--Roch theorem in algebraic geometry.
For the purposes of our multiplication algorithm, we will pick some divisor $D = \sum n_P P$ such that $\rm{supp}(D) = \set{P \; : \; n_P \neq 0}$ is disjoint from $\cal P \cup \set{Q}$.
Since the product of elements $f,g \in \cal L(D)$ is an element $fg \in \cal L(2D)$, both $\cal L(D)$ and $\cal L(2D)$ will be relevant to us.

Through the evaluation map, given a field element $f \in F$ and a place $P$ of degree $d$, we are able to extract a single number $f(P) \in \GF{2^d}$. 
However, in polynomial interpolation, it is quite common to obtain multiple values per point through evaluation of the function's derivatives.
We can employ a similar idea here.
Let $t_P \in P$ be a local parameter of $P$, then $t_P$ is such that $P = t_P\cal O_P$ and $\bigcap_{n = 1}^\infty t_P^n \cal O_P = \set{0}$.
Elements $f \in \cal O_P$ can be expanded in the form
\begin{align}\label{eq:local-expansion}
    f = f(P) + f'(P) t_P + f''(P)t_P^2 + \dots
\end{align}
with the coefficients $f^{(k)}(P) \in F_P$ playing the same role as derivatives.
This gives rise to the notion of a generalized evaluation map.
Given a divisor $D$, a place $P$ not in $\rm{supp}(D)$, and an integer multiplicity $u \geq 1$, we define the \emph{generalized evaluation at} $P$ to be the function 
\begin{equation}\label{eq:generalized-evaluation}
    \begin{aligned}
        \varphi_{D,P,u} : \mathcal{L}(D) & \to F_P^{u}\\
        f & \mapsto (f(P),f'(P),\dots, f^{(u-1)}(P)).
    \end{aligned}
\end{equation}
This can be extended to the case $P \in \rm{supp}(D)$ by first multiplying $f$ by $t_P^{n_P}$, where $n_P$ is the coefficient of $P$ in $D$; since $v_P(f) \geq -n_P$ for $f \in \cal L(D)$, the product $t_P^{n_P} f$ lies in $\cal O_P$.

\subsection{Chudnovsky--Chudnovsky Multiplication Algorithm---Optimal versus Practical}
\label{sec:projective-line-et-al}

In their original paper \cite{ChudnovskyChudnovsky1988}, the Chudnovsky brothers introduced this method, using a degree $n$ place, $Q$, and a large number of degree one places $\set{P_1,\dots,P_N}$, and argued that it leads to bilinear complexity $O(n)$.
To achieve this, they considered function fields arising from modular curves, which had the drawback of only working for $\GF{q^n}$ with $q$ a sufficiently large square, and their argument was completed later~\cite{STV1992,Cascudo2012}. Ballet~\cite{BALLET1999364} gave the first complete proof of linear growth for all $q$, using the so-called Garcia--Stichtenoth tower of function fields~\cite{GarciaStichtenoth1995}, and later work decreased the constant using places of higher degree~\cite{Ballet2021survey}.

These results are stunning, but they are not practical.
It would be unfair to call the results non-constructive; they can be used to construct $A,B,C$ as in Equation~\eqref{eq:bilinear_decomposition}, but until recently there were no known polynomial-time constructions of the algorithms they describe~\cite{Ballet2021survey}.
This is due to the particular difficulty of finding a degree $n$ place and also a suitable divisor in the function fields that they consider.
Once these two items are supplied, the rest of the construction is relatively straightforward.
A polynomial time algorithm for computing $A,B,C$ achieving $O(n)$ bilinear complexity was demonstrated in \cite{BBP2022}. However, its guaranteed constant is large (for $q = 2$, the bound is $648n$), and the machinery that makes it linear only matters asymptotically.
For the relatively small values of $n$ that we care about in quantum computing, a simpler variant suffices: the polynomial time construction over the projective line detailed in \cite{BBP2020projective}, which we now describe.
Here we take our function field to be the field of rational functions $F = \GF{2}(x)$.
This field has three degree one places $P_0, P_1, P_\infty$ corresponding to the three points in the projective line $\mathbb{P}^1$, and the degree $d > 1$ places of $F$ correspond to degree $d$ irreducible polynomials in $\GF{2}[x]$.

The construction proceeds as follows.

For small values of $n$, we find a degree $n$ irreducible polynomial in $\GF{2}[x]$, and thus a degree $n$ place $Q$, by brute-force search.
For larger values of $n$, a more efficient algorithm such as the algorithm of Shoup~\cite{Shoup1990} can be used.
For our divisor, we pick $D = (n-1) P_\infty$ giving us the Riemann--Roch spaces $\cal L(D) = \set{q \in \GF{2}[x] \; : \; \deg(q) \leq n - 1}$, since a rational function having a pole of order at most $n-1$ at $P_\infty$ is exactly a polynomial of degree at most $n - 1$, and $\cal L(2D) = \set{q \in \GF{2}[x] \; : \; \deg(q) \leq 2n - 2}$.
Lastly, we select a sufficient number of places, for instance in order of increasing degree~\cite[Section~4.2]{BP2023generalized}, or by the selection of \cite[Section~5.1]{BBP2020projective}, for which the asymptotic bound below is proved.
This yields a collection of places the degrees of which, up to multiplicity, sum to $2n - 1$.
This matches the dimension of the space $\cal L(2D)$ and corresponds exactly with the interpolation results discussed previously.
Interestingly, this includes selecting the place $P_\infty$, which seemingly contradicts the assumption that $D$ has support disjoint from $\cal P \cup \set{Q}$.
This, however, is handled cleanly by generalized evaluation \eqref{eq:generalized-evaluation}.

Multiplying elements $q,r \in \cal L(D)$ then reduces to multiplying the values that they take at the low-degree places $\cal P$, with low-degree here meaning $d = O(\log n)$.
The values that $q$ and $r$ take at a place of degree $d$ are elements, up to isomorphism, of $\GF{2^d}$.
For small $d$, we can compute an optimal multiplication algorithm by brute-force, but for $d$ larger we instead recurse, applying the same construction above to work out an efficient multiplication algorithm in $\GF{2^d}$.
In this recursive form, the $t$-th layer of recursion requires places of degree $O(\log^{(t)}n)$.
This recursion yields a quasi-linear bilinear complexity, $O(n\,16^{\log^*_{\sqrt 2}(2n)})$ for $q=2$~\cite[Theorem~3]{BBP2020projective}, but for the values of $n$ that we care about this is perfectly acceptable.
Our implementation, whose lengths $\mu(n)$ are shown in Fig.~\ref{fig:mu}, follows this recursive construction with two minor differences: it selects places in order of increasing degree without using $P_\infty$, and it handles some of the places by the analogous construction over $\GF{4}$, as in the hybrid algorithm of \cite{BBP2022}.

\section{Addition in $O(1)$ Toffoli depth}
\label{sec:bilinear_to_circuit}

We now use the bilinear decomposition of the preceding sections to
construct quantum circuits for encoded addition. We first give a
fully unitary construction and then show how measurement-assisted
uncomputation reduces its Toffoli depth.

\subsection{Reversible constructions of encoded addition}

We proceed as follows. We first show that encoded out-of-place addition can be realized in constant Toffoli depth. Then we show that in-place addition can also be realized in constant Toffoli depth, using the out-of-place addition.

\begin{proposition}[Out-of-place addition]
\label{prop:oop_addition}
Suppose multiplication in $\GF{2^n}$ is given by a bilinear
decomposition of length $\mu(n)$. Then, for any
$\alpha,\beta\in\mathbb Z_{2^n-1}$ and any $h\in\GF{2^n}$,
\begin{equation}
\label{eq:encoded_oop_addition}
    |\mathcal E_0(\alpha)\rangle
    |\mathcal E_0(\beta)\rangle
    |h\rangle
    \longmapsto
    |\mathcal E_0(\alpha)\rangle
    |\mathcal E_0(\beta)\rangle
    |h\oplus\mathcal E_0(\alpha+\beta)\rangle
\end{equation}
can be implemented with Toffoli depth $2$, using
$2\mu(n)$ Toffoli gates and $3\mu(n)$ workspace qubits.
The total number of qubits is therefore $3n+3\mu(n)$.
In particular, if $\mu(n)=O(n)$, both the Toffoli count and the total
space are $O(n)$.
\end{proposition}

\begin{proof}
Compute $Au$ and $Bv$ using CNOT gates, where
$u=\mathcal E_0(\alpha)$ and $v=\mathcal E_0(\beta)$.
The entries of $(Au)\circ(Bv)$ are computed in parallel using one
Toffoli layer. The linear map $C$ is then applied using CNOT gates to
add
\begin{equation}
    C((Au)\circ(Bv))=uv
\end{equation}
to the target register. A second application of the same Toffoli layer
clears the temporary products, after which the CNOTs used in the computation of $Au$ and $Bv$ are reversed. Since $uv=f^\alpha f^\beta
      =f^{\alpha+\beta}
      =\mathcal E_0(\alpha+\beta),$ this gives the desired out-of-place addition and returns all workspace qubits to $|0\rangle$.
\end{proof}

To obtain an in-place circuit, one must additionally remove $\mathcal{E}_0(\beta)$. For this purpose, we introduce the \emph{inverse-pair encoding}
\begin{equation}
    \mathcal E_p(\alpha)
    =
    \bigl(f^\alpha,f^{-\alpha}\bigr).
\end{equation}
Under this encoding, the desired transformation is
\begin{equation}
    (f^\alpha,f^{-\alpha},f^\beta,f^{-\beta})
    \longmapsto
    (f^\alpha,f^{-\alpha},
     f^{\alpha+\beta},f^{-(\alpha+\beta)}).
\end{equation}
The main advantage of this encoding is that the data that needs to be erased is readily computable from the output in the preceding equation as follows:
\begin{equation}
\label{eq:ip_cleanup_identities}
    f^\beta=f^{-\alpha}f^{\alpha+\beta},
    \qquad
    f^{-\beta}=f^\alpha f^{-(\alpha+\beta)}.
\end{equation}
All in-place circuits below act as stated on valid encodings, i.e., on the span of $\{|\mathcal E_p(\alpha)\rangle|\mathcal E_p(\beta)\rangle\}$.

\begin{theorem}[In-place addition]
\label{thm:depth_four_inplace}
Under the encoding $\mathcal E_p$, exact in-place addition can be
implemented with Toffoli depth $4$, using $8\mu(n)$ Toffoli gates and
$2n+6\mu(n)$ workspace qubits, for a total of
$6n+6\mu(n)$ qubits. In particular, if $\mu(n)=O(n)$, both the
Toffoli count and the total space are $O(n)$.
\end{theorem}

\begin{proof}
First compute $f^{\alpha+\beta}$ and
$f^{-(\alpha+\beta)}$ in parallel in two fresh ancillary registers using
two copies of Proposition~\ref{prop:oop_addition}. Swap these two registers
with the registers holding $\mathcal E_p(\beta)$. Equation~\eqref{eq:ip_cleanup_identities}
then allows the registers containing $f^\beta$ and $f^{-\beta}$ to be
cleared using two further multiply-add operations, again in parallel.

The construction therefore consists of two sequential rounds of
Toffoli depth $2$. The four multiply-add operations use
$8\mu(n)$ Toffoli gates. In addition to the $4n$ data qubits, the
construction uses two $n$-qubit ancillary registers and two parallel
multiplier workspaces, giving $2n+6\mu(n)$ workspace qubits.
\end{proof}

It is evident from these constructions that some of the Toffolis are used for uncomputation. One can remove these using measurement-based uncomputation~\cite{Gidney2018,Jones2013}, discussed below. We actually prove a slightly more general result, which may be more broadly useful.

\begin{lemma}[Bilinear measurement erasure]
\label{lem:bilinear_erasure}
Let $\Phi:\mathbb F_2^p\times\mathbb F_2^q
    \longrightarrow\mathbb F_2^r$ be a bilinear map and consider the following operation:
    \begin{equation}
        |x\rangle|y\rangle|0\rangle \to |x\rangle|y\rangle |\Phi(x,y)\rangle.
    \end{equation}
The register storing $\Phi(x,y)$ can be uncomputed using single-qubit $X$-basis measurements and classically controlled CZ corrections, without using any Toffoli gates.
\end{lemma}

\begin{proof}
Measuring the output register in the $X$ basis with outcome
$m\in\mathbb F_2^r$ introduces the phase $(-1)^{m\cdot\Phi(x,y)}$. Since $\Phi$ is bilinear, there is a binary matrix $Q_m$ such that
\begin{equation}
    m\cdot\Phi(x,y)=x^TQ_my.
\end{equation}
Applying a CZ gate between the corresponding qubits of $x$ and $y$
for each nonzero entry of $Q_m$ removes this phase. The measured
register can then be reset to $|0\rangle$. By linearity, the procedure
is applicable to arbitrary superpositions.
\end{proof}

Applying Lemma~\ref{lem:bilinear_erasure} to the temporary
coordinatewise products removes the second Toffoli layer of the
preceding construction.

\begin{corollary}[Depth-one out-of-place addition]
\label{cor:depth_one_oop}
The out-of-place encoded addition of
Proposition~\ref{prop:oop_addition} can be implemented with Toffoli depth
$1$, using $\mu(n)$ Toffoli gates and $3\mu(n)$ workspace qubits.
The total number of qubits is therefore $3n+3\mu(n)$.
The circuit uses single-qubit measurements and classically controlled
Clifford corrections.
\end{corollary}

\begin{proof}
Compute the coordinatewise products with one Toffoli layer and add
their contribution to the target as in
Proposition~\ref{prop:oop_addition}. Instead of applying a second Toffoli
layer, erase the product register using
Lemma~\ref{lem:bilinear_erasure}. Finally, reverse the CNOT
computations of $Au$ and $Bv$ to return the remaining workspace qubits
to $|0\rangle$.
\end{proof}

The same erasure procedure also removes the second multiplication
round of the in-place construction. After
$f^{\alpha+\beta}$ and $f^{-(\alpha+\beta)}$ have been computed,
Eq.~\eqref{eq:ip_cleanup_identities} expresses each component of
$\mathcal E_p(\beta)$ as a bilinear function of registers that remain.
It can therefore be erased without any additional
Toffoli gates.

\begin{theorem}[Depth-one in-place encoded addition]
\label{thm:main_inplace}
Under the encoding $\mathcal E_p$, exact in-place addition can be
implemented with Toffoli depth $1$, using $2\mu(n)$ Toffoli gates and
$2n+6\mu(n)$ workspace qubits, for a total of
$6n+6\mu(n)$ qubits. The circuit uses single-qubit measurements and
classically controlled Clifford corrections.
\end{theorem}

\begin{proof}
Starting from
$|\mathcal E_p(\alpha)\rangle|\mathcal E_p(\beta)\rangle$,
compute $f^{\alpha+\beta}$ and $f^{-(\alpha+\beta)}$ in two fresh
ancillary registers using two copies of
Corollary~\ref{cor:depth_one_oop}. The two multiplications act on
disjoint registers, so their Toffoli layers can be executed in
parallel.

By Eq.~\eqref{eq:ip_cleanup_identities}, the registers containing
$f^\beta$ and $f^{-\beta}$ can then be erased using
Lemma~\ref{lem:bilinear_erasure}, with no additional Toffoli gates.
Swapping the two newly computed registers with the cleared registers
that held $\mathcal E_p(\beta)$ yields
\begin{equation}
    |\mathcal E_p(\alpha)\rangle
    |\mathcal E_p(\alpha+\beta)\rangle,
\end{equation}
while returning the two additional ancillary registers to $|0\rangle$.

The two parallel multiplications use $2\mu(n)$ Toffoli gates in a
single layer. The $4n$ data qubits are supplemented by two $n$-qubit
ancillary registers and two parallel multiplier workspaces, giving
$2n+6\mu(n)$ workspace qubits and $6n+6\mu(n)$ qubits in total.
\end{proof}

\begin{remark}
Since bilinear decompositions with $\mu(n)=O(n)$ can be constructed
deterministically in time $O(n^4)$~\cite[Theorem~4.1]{BBP2022},
Theorem~\ref{thm:main_inplace} yields an in-place adder with Toffoli depth $1$,
$O(n)$ Toffoli gates, and $O(n)$ total qubits.
\end{remark}

\section{Space--depth tradeoffs}
\label{sec:optimized}

In this section, we focus on reducing the constants in our construction, which are dominated by $\mu(n)$. Known upper bounds on $\mu(n)$ for small $n$ are listed in Table~\ref{tab:mu_bound}, and for $2\le n\le 40$ our implementation of the method in Section~\ref{sec:projective-line-et-al} yields $2.75n \le \mu(n) \le 6.44n$ (Fig.~\ref{fig:mu}). Plugging these numbers into Theorem~\ref{thm:main_inplace} yields roughly $20n$--$40n$ workspace qubits, which can be large in practice. Fortunately, this cost can be reduced substantially using two tricks: (i) evaluating the bilinear terms in batches and (ii) reusing ancillary registers between successive multiplications. Both reduce the space by a constant factor at the cost of increasing the Toffoli depth by a related constant factor.

\begin{table}[h]
\centering
\small
\begin{tabular}{l|ccccccccccccccccc}
$n$ & 2 & 3 & 4 & 5 & 6 & 7 & 8 & 9 & 10 & 11 & 12 & 13 & 14 & 15 & 16 & 17 & 18 \\
\hline
$\mu(n) \le$ & 3 & 6 & 9 & 13 & 15 & 22 & 24 & 30 & 33 & 39 & 42 & 48 & 51 & 54 & 60 & $67^{\dagger}$ & 69 
\end{tabular}
\caption{Explicit upper bounds on the bilinear complexity for the symmetric case (i.e., $A=B$) from Refs.~\cite{CenkOzbudak2009,Ballet2021survey,BalletPieltant2018}. $^{\dagger}$We could not locate an explicit length-$67$ decomposition for $n=17$; Ref.~\cite{CenkOzbudak2009} establishes $\mu(17)\le 68$.}
\label{tab:mu_bound}
\end{table}

\begin{figure}[h]
\centering
\begin{subfigure}[b]{0.48\textwidth}
\centering
\includegraphics[width=\textwidth]{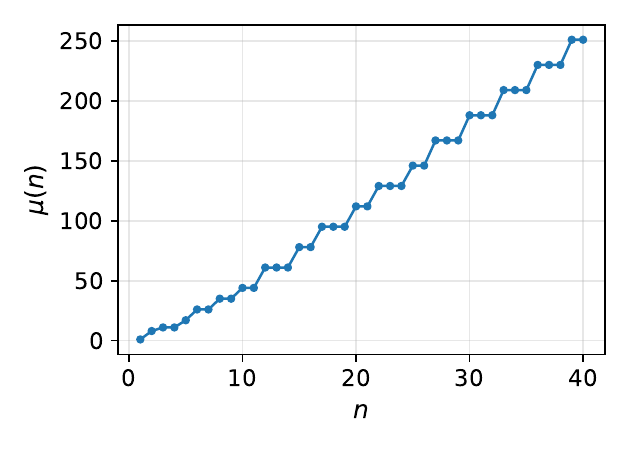}
\caption{}
\label{fig:mu_a}
\end{subfigure}
\hfill
\begin{subfigure}[b]{0.48\textwidth}
\centering
\includegraphics[width=\textwidth]{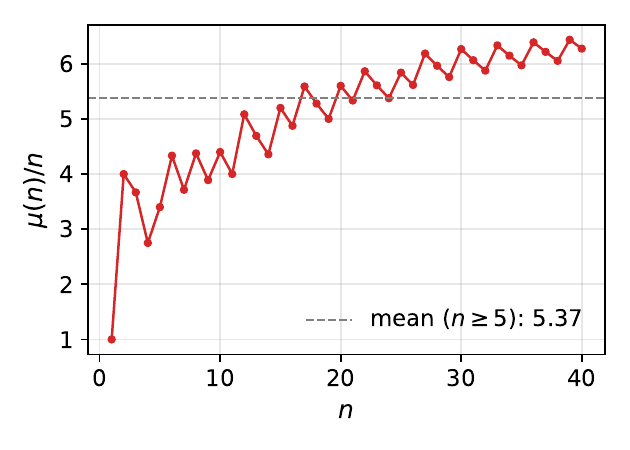}
\caption{}
\label{fig:mu_b}
\end{subfigure}
\caption{(a) Length $\mu(n)$ of the bilinear decomposition produced by our implementation of the recursive Chudnovsky-type construction of Section~\ref{sec:projective-line-et-al}, for $n \le 40$. (b) The ratio $\mu(n)/n$; the dashed line is its mean over $5\le n\le 40$.}
\label{fig:mu}
\end{figure}

\subsection{Workspace reduction for field multiplication}
\label{sec:optimized_multiplier}

The encoded out-of-place addition in Eq.~\eqref{eq:encoded_oop_addition}
uses three $n$-qubit data registers. We first show that these $3n$ qubits suffice to implement the operation in constant Toffoli depth when $\mu(n)=O(n)$.
We then establish a space--depth tradeoff by batching the terms of a fixed bilinear decomposition.

The idea is to divide the bilinear terms into groups of at most
$n$ terms each. If the three lists of vectors $\{a_i\}$, $\{b_i\}$, and $\{c_i\}$
within a group are each linearly independent, CNOT basis changes allow the group's contribution to be accumulated directly in the output register with one Toffoli layer and no workspace.
When this independence condition fails, we split some of the bilinear terms to obtain at most eight groups, each satisfying the condition.

\begin{theorem}[Zero-workspace out-of-place addition]
\label{thm:workspace_free_multiplication}
Given a bilinear decomposition of length $\mu(n)$, the out-of-place addition in Eq.~\eqref{eq:encoded_oop_addition} can be implemented using only CNOT and Toffoli gates, with zero workspace qubits and at most $8\mu(n)$ Toffoli gates, with Toffoli depth
\begin{equation}
\label{eq:workspace_free_depth}
D_{\mathrm{Toffoli}}\le 8\left\lceil\frac{\mu(n)}{n}\right\rceil.
\end{equation}
The circuit uses $3n$ qubits in total and can be constructed in polynomial time. In particular, its Toffoli depth is $O(1)$ whenever $\mu(n)=O(n)$.
\end{theorem}

\begin{proof}
Partition the $\mu(n)$ bilinear terms into
$\lceil\mu(n)/n\rceil$ groups of size at most $n$. We first split each group further so that, within each resulting group, the vectors $a_i$ are linearly independent, and likewise for $b_i$ and $c_i$. For a group $S$, let $\{a_i:i\in I\}$ be a maximal linearly independent sublist and let $J=S\setminus I$. Since $|S|\le n$, basis extension gives vectors $w_j\in\Ftwo^n$, for $j\in J$, such that $\{a_i:i\in I\}\cup\{w_j:j\in J\}$ is linearly independent. For each $j\in J$, split the corresponding term as
\begin{equation}
c_j(a_j\cdot u)(b_j\cdot v)
=c_j(w_j\cdot u)(b_j\cdot v)+c_j((a_j+w_j)\cdot u)(b_j\cdot v),
\end{equation}
and place the terms indexed by $I$ together with the first terms on the right-hand side in one group, and the second terms in another. The first group is linearly independent by construction, and so is $\{a_j+w_j:j\in J\}$, since each $a_j$ lies in the span of $\{a_i:i\in I\}$. Repeating this splitting for $b_i$ and then $c_i$, each of which at most doubles the number of groups, gives
\begin{equation}
\sum_{i\in S}c_i(a_i\cdot u)(b_i\cdot v)
=
\sum_{\ell=1}^{r}\sum_{i=1}^{s_\ell}
\tilde c_{\ell,i}
(\tilde a_{\ell,i}\cdot u)(\tilde b_{\ell,i}\cdot v),
\end{equation}
where $r\le 8$, $s_\ell\le n$, and, for each $\ell$, the vectors $\tilde a_{\ell,i}$ are linearly independent, and likewise for $\tilde b_{\ell,i}$ and $\tilde c_{\ell,i}$.

Now fix one of these groups, omitting the index $\ell$, and write its size as $s$. By basis extension, choose invertible matrices $P,Q,R$ such that, for each $i\le s$, the $i$th rows of $P,Q$ are $\tilde a_i^T,\tilde b_i^T$, respectively, and the $i$th column of $R$ is $\tilde c_i$. Apply the CNOT transformations
$(u,v,h)\mapsto(Pu,Qv,R^{-1}h)$, and then, for each $i\le s$, a Toffoli gate controlled on the $i$th coordinates of $Pu$ and $Qv$ and targeting the $i$th output coordinate. These gates act on disjoint qubits and form one layer. Applying the inverse CNOT transformations $P^{-1},Q^{-1},R$ returns the inputs to $u,v$, while the output becomes
\begin{equation}
R\left(R^{-1}h+\sum_{i=1}^{s}e_i(Pu)_i(Qv)_i\right)
=
h+\sum_{i=1}^{s}\tilde c_i
(\tilde a_i\cdot u)(\tilde b_i\cdot v),
\end{equation}
where $e_i$ is the $i$th standard basis vector.

Processing all resulting groups sequentially accumulates their contributions to give $h\mapsto h+uv$. The circuit uses only the three data registers, with at most $8\lceil\mu(n)/n\rceil$ Toffoli layers and $8\mu(n)$ Toffoli gates. All independent sublists, basis extensions, and CNOT circuits can be constructed by Gaussian elimination in time polynomial in $n$ and $\mu(n)$.
\end{proof}

The preceding construction guarantees zero-workspace
out-of-place addition in constant Toffoli depth when
$\mu(n)=O(n)$, but may increase the number of bilinear terms
by a factor of up to eight. To reduce this overhead, we keep the original bilinear
terms and evaluate them in batches, reusing workspace
between successive batches. Here, a \emph{batch} is a subset
$S\subseteq\{1,\ldots,\mu(n)\}$ of the bilinear terms in
Eq.~\eqref{eq:bilinear_decomposition} that are evaluated together in one Toffoli layer.

The workspace required by a batch depends on the linear independence of its input forms and output directions. As in the preceding construction, independent forms and directions can be placed in the existing data registers by CNOT basis changes.

For a batch $S$, let $A_S,B_S$ contain the rows of $A,B$
indexed by $S$, and let $C_S$ contain the corresponding
columns of $C$. Define
\begin{equation}
\label{eq:batch_ranks}
r_A(S)=\operatorname{rank}(A_S),\qquad
r_B(S)=\operatorname{rank}(B_S),\qquad
r_C(S)=\operatorname{rank}(C_S).
\end{equation}

Each input form or output direction that is not part of an independent set of maximal size requires one workspace qubit. Hence a batch $S$ requires
\begin{equation}
\label{eq:batch_workspace}
\kappa(S)
=(|S|-r_A(S))+(|S|-r_B(S))+(|S|-r_C(S))
\end{equation}
workspace qubits.

For a partition into $L$ nonempty batches
$S_1,\ldots,S_L$, the same workspace is reused between batches,
so its peak size is $\max_\ell\kappa(S_\ell)$.
We therefore define
\begin{equation}
\label{eq:kappa_L}
\kappa(L)
=\min_{(S_1,\ldots,S_L)}
\max_{1\le\ell\le L}\kappa(S_\ell),
\end{equation}
where the minimum is over all such partitions of
$\{1,\ldots,\mu(n)\}$.

\begin{theorem}[Batched out-of-place addition]
\label{thm:batched_multiplication}
For every $1\le L\le\mu(n)$, encoded out-of-place addition
can be implemented using $\mu(n)$ Toffoli gates,
Toffoli depth at most $L$, and $\kappa(L)$ workspace qubits,
for a total of $3n+\kappa(L)$ qubits.
The implementation uses single-qubit measurements
and classically controlled Clifford corrections.
\end{theorem}

\begin{proof}
Consider a batch $S$ of size $s$. As in the proof of Theorem~\ref{thm:workspace_free_multiplication}, use CNOT basis changes to place $r_A(S)$ independent input forms in the first data register and $r_B(S)$ in the second, and compute the remaining forms into $s-r_A(S)$ and $s-r_B(S)$ workspace qubits, respectively. Similarly, choose an output basis that maps $r_C(S)$ independent columns of $C_S$ to distinct standard basis vectors. The corresponding products are accumulated directly in the output register, and the other $s-r_C(S)$ products are computed into workspace. All $s$ Toffoli gates act on disjoint triples and form one layer. After adding the temporary products' contributions to the output using CNOT gates, erase those products using Lemma~\ref{lem:bilinear_erasure}, and reverse the computations of the additional input forms and all basis changes. This restores the inputs and returns all $\kappa(S)$ workspace qubits to zero.

Processing the batches sequentially uses at most $L$
Toffoli layers and evaluates each of the $\mu(n)$ terms once.
Choosing a partition attaining Eq.~\eqref{eq:kappa_L}
gives the stated workspace bound.
\end{proof}

When $\kappa(L)=0$, the circuit uses only the $3n$
data qubits. Every batch must then have independent
input forms and output directions, so it contains
at most $n$ terms. Hence zero workspace requires
\begin{equation}
\label{eq:min_batch_number}
L\ge L_0:=\left\lceil\frac{\mu(n)}{n}\right\rceil.
\end{equation}

A partition into $L$ batches is called \emph{balanced}
if the batch sizes differ by at most one.
If there exists a balanced partition in which every batch $S$
satisfies $r_A(S)=r_B(S)=r_C(S)=\min\{|S|,n\}$, then
\begin{equation}
\label{eq:depth_space_rel}
\kappa(L)
=
3\max\left\{
\left\lceil\frac{\mu(n)}{L}\right\rceil-n,\,
0
\right\}.
\end{equation}
Conversely, the right-hand side of Eq.~\eqref{eq:depth_space_rel} is a lower bound on $\kappa(L)$ for any partition into $L$ batches, since some batch contains at least $\lceil\mu(n)/L\rceil$ terms and every rank is at most $n$.

Theorem~\ref{thm:workspace_free_multiplication} guarantees zero workspace at depth at most $8L_0$ through term splitting. For the original decomposition, a small increase above $L_0$ may suffice in practice.
To motivate this, consider balanced partitions into
$L=L_0+c\le n$ batches, where $c$ is a nonnegative integer. Each batch size is at most $\lceil\mu(n)/L\rceil\le n-c$, because $\mu(n)/L\le nL_0/L=n-nc/L\le n-c$. As a heuristic, model the batch submatrices
$A_S,B_S,C_S^T$ as uniformly random binary matrices.
A random $(n-c)\times n$ binary matrix has full row rank
with probability
\begin{equation}
    p_{n,c}=\prod_{j=c+1}^{n}(1-2^{-j}),
\end{equation}
so its rank-deficiency probability is at most
$1-p_{n,c}\le 2^{-c}$.
Since each batch contains at most $n-c$ terms,
a union bound over the three submatrices of all $L$ batches
gives a failure probability at most $3L2^{-c}$ under this assumption.
This suggests that $L_0+O(1)$ batches may suffice in practice
to obtain a zero-workspace partition. In fact, this heuristic turns out to be pessimistic: for every decomposition we tested, $L_0$ batches already suffice, i.e., $\kappa(L_0)=0$ (Table~\ref{tab:frontier_kappa}).

For $2\le n\le18$, we verified explicit symmetric bilinear
decompositions matching the bounds in Table~\ref{tab:mu_bound} for all $n\neq 17$.
We computed their space--depth tradeoffs, obtaining the
workspace costs in Table~\ref{tab:frontier_kappa}.
Every entry attains the lower bound in Eq.~\eqref{eq:depth_space_rel}, so these values are optimal for the given decompositions.

\begin{table}[H]
\centering
\normalsize
\setlength{\tabcolsep}{4pt}
\begin{tabular}{l|*{17}{r}}
\hline
$n$ & 2 & 3 & 4 & 5 & 6 & 7 & 8 & 9 & 10 & 11 & 12 & 13 & 14 & 15 & 16 & 17 & 18 \\
$\mu(n)$ & 3 & 6 & 9 & 13 & 15 & 22 & 24 & 30 & 33 & 39 & 42 & 48 & 51 & 54 & 60 & $68^{\dagger}$ & 69 \\
\hline
$\kappa(1)$ & 3 & 9 & 15 & 24 & 27 & 45 & 48 & 63 & 69 & 84 & 90 & 105 & 111 & 117 & 132 & 153 & 153 \\
$\kappa(2)$ & 0 & 0 & 3 & 6 & 6 & 12 & 12 & 18 & 21 & 27 & 27 & 33 & 36 & 36 & 42 & 51 & 51 \\
$\kappa(3)$ & 0 & 0 & 0 & 0 & 0 & 3 & 0 & 3 & 3 & 6 & 6 & 9 & 9 & 9 & 12 & 18 & 15 \\
$\kappa(4)$ & -- & 0 & 0 & 0 & 0 & 0 & 0 & 0 & 0 & 0 & 0 & 0 & 0 & 0 & 0 & 0 & 0 \\
\hline
\end{tabular}
\caption{Exact workspace cost $\kappa(L)$ for the explicit symmetric bilinear decompositions used in our finite-size study. A dash indicates $L>\mu(n)$. For $n=17$, Refs.~\cite{Ballet2021survey,BalletPieltant2018} report $\mu(17)\le 67$, but we could not reproduce a length-$67$ decomposition; we therefore use a verified length-$68$ decomposition.}
\label{tab:frontier_kappa}
\end{table}

\subsection{Space-efficient in-place addition}
\label{sec:black_box_adder}

We now use the batched multiplier with Toffoli depth $L$, Toffoli
count $\mu(n)$, and workspace $\kappa(L)$. The construction of
Theorem~\ref{thm:main_inplace} computes
$f^{\alpha+\beta}$ and $f^{-(\alpha+\beta)}$ simultaneously.
We next consider two alternatives that reduce workspace by increasing Toffoli depth.

The two output components need not be computed simultaneously. One can
instead compute the first component, erase the corresponding old input,
and reuse that register for the second component.

\begin{proposition}
\label{lem:two_multiplication}
Under the encoding $\mathcal E_p$, exact in-place addition can be
implemented using $2\mu(n)$ Toffoli gates, Toffoli depth $2L$, and
$n+\kappa(L)$ workspace qubits, for a total of
$5n+\kappa(L)$ qubits. The circuit uses single-qubit measurements and
classically controlled Clifford corrections.
\end{proposition}

\begin{proof}
Starting from the four data registers and one blank ancillary register,
first compute
\begin{equation}
    (f^\alpha,f^{-\alpha},f^\beta,f^{-\beta},0)
    \longmapsto
    (f^\alpha,f^{-\alpha},f^\beta,f^{-\beta},
     f^{\alpha+\beta}).
\end{equation}
Using
$f^\beta=f^{-\alpha}f^{\alpha+\beta}$,
erase the register containing $f^\beta$ with
Lemma~\ref{lem:bilinear_erasure}. The cleared register is then reused
to compute $f^{-(\alpha+\beta)}$. Finally,
$f^{-\beta}=f^\alpha f^{-(\alpha+\beta)}$ allows the remaining old
input register to be erased in the same way.

The two field multiplications are sequential and reuse the same
multiplier workspace. Hence the circuit uses $2\mu(n)$ Toffoli gates,
Toffoli depth $2L$, and one additional $n$-qubit ancillary register
together with $\kappa(L)$ multiplier workspace.
\end{proof}

The additional ancillary register can be eliminated altogether by allowing
one more multiplication. For compactness, write
$u=f^\alpha$ and $v=f^\beta$ within this construction.

\begin{proposition}
\label{lem:three_multiplication}
Under the encoding $\mathcal E_p$, exact in-place addition can be
implemented using $3\mu(n)$ Toffoli gates, Toffoli depth $3L$, and
$\kappa(L)$ workspace qubits, for a total of
$4n+\kappa(L)$ qubits. (When $\kappa(L)>0$, the circuit also uses single-qubit measurements and classically controlled Clifford corrections.)
\end{proposition}

\begin{proof}
Apply the following three multiply-add operations, where $M_{i,j\to k}$ adds the product of registers $i$ and $j$ to register $k$:
\begin{align}
(u,u^{-1},v,v^{-1})
&\xrightarrow{\,M_{1,3\to4}\,}
(u,u^{-1},v,v^{-1}+uv)
\nonumber\\
&\xrightarrow{\,M_{2,4\to3}\,}
(u,u^{-1},u^{-1}v^{-1},v^{-1}+uv)
\nonumber\\
&\xrightarrow{\,M_{1,3\to4}\,}
(u,u^{-1},u^{-1}v^{-1},uv).
\end{align}
The second step follows from
\begin{equation}
    v+u^{-1}(v^{-1}+uv)=u^{-1}v^{-1},
\end{equation}
and the third from
\begin{equation}
    v^{-1}+uv+u(u^{-1}v^{-1})=uv.
\end{equation}
Swapping the final two registers gives
\begin{equation}
    (u,u^{-1},uv,u^{-1}v^{-1}),
\end{equation}
which is the desired encoded sum. The three multiplications are
sequential and reuse the same workspace.
\end{proof}

\subsection{Resource comparison}
\label{sec:optimized_comparison}

We first compare our three encoded in-place constructions, summarized in Table~\ref{tab:adder_resource_comparison}. The parallel construction minimizes Toffoli depth, while the other two reduce workspace by reusing registers across successive multiplications at the cost of increased Toffoli depth. Increasing the number of batches provides an additional space–depth tradeoff.

\begin{table}[htbp]
\centering
\small
\begin{tabular}{l|c|c|c|c}
\hline
Construction & Toffoli count & Toffoli depth & Workspace qubits & Total qubits\\
\hline
Parallel two multiplications
& $2\mu(n)$
& $L$
& $2n+2\kappa(L)$
& $6n+2\kappa(L)$\\
Sequential two multiplications
& $2\mu(n)$
& $2L$
& $n+\kappa(L)$
& $5n+\kappa(L)$\\
Three multiplications
& $3\mu(n)$
& $3L$
& $\kappa(L)$
& $4n+\kappa(L)$\\
\hline
\end{tabular}
\caption{Space--depth tradeoffs for encoded in-place addition. The four encoded data registers account for $4n$ qubits in every construction.}
\label{tab:adder_resource_comparison}
\end{table}

To put these space--depth tradeoffs in perspective, we compare our encoded constructions with existing binary adders from the literature~\cite{Takahashi2010Addition,Gidney2018,Remaud2025,Thapliyal2020QCLA,Gidney2020Blocks,Wang2024OptimalDepth,Gu2025Adder}. We consider seven designs, ranging from space-efficient ripple-carry adders to depth-efficient carry-lookahead adders. For constructions originally designed for out-of-place addition, we use the standard conversion to in-place addition~\cite{Gidney2020Blocks} together with Gidney's measurement-based uncomputation~\cite{Gidney2018}. This conversion does not require any additional Toffoli gates or workspace, preserving the Toffoli depth and total qubit count of the original circuits. Since our encoding uses $4n$ data qubits, compared with $2n$ for binary addition, we compare total qubit counts.

\begin{figure}[htbp]
\centering
\includegraphics[width=\textwidth]{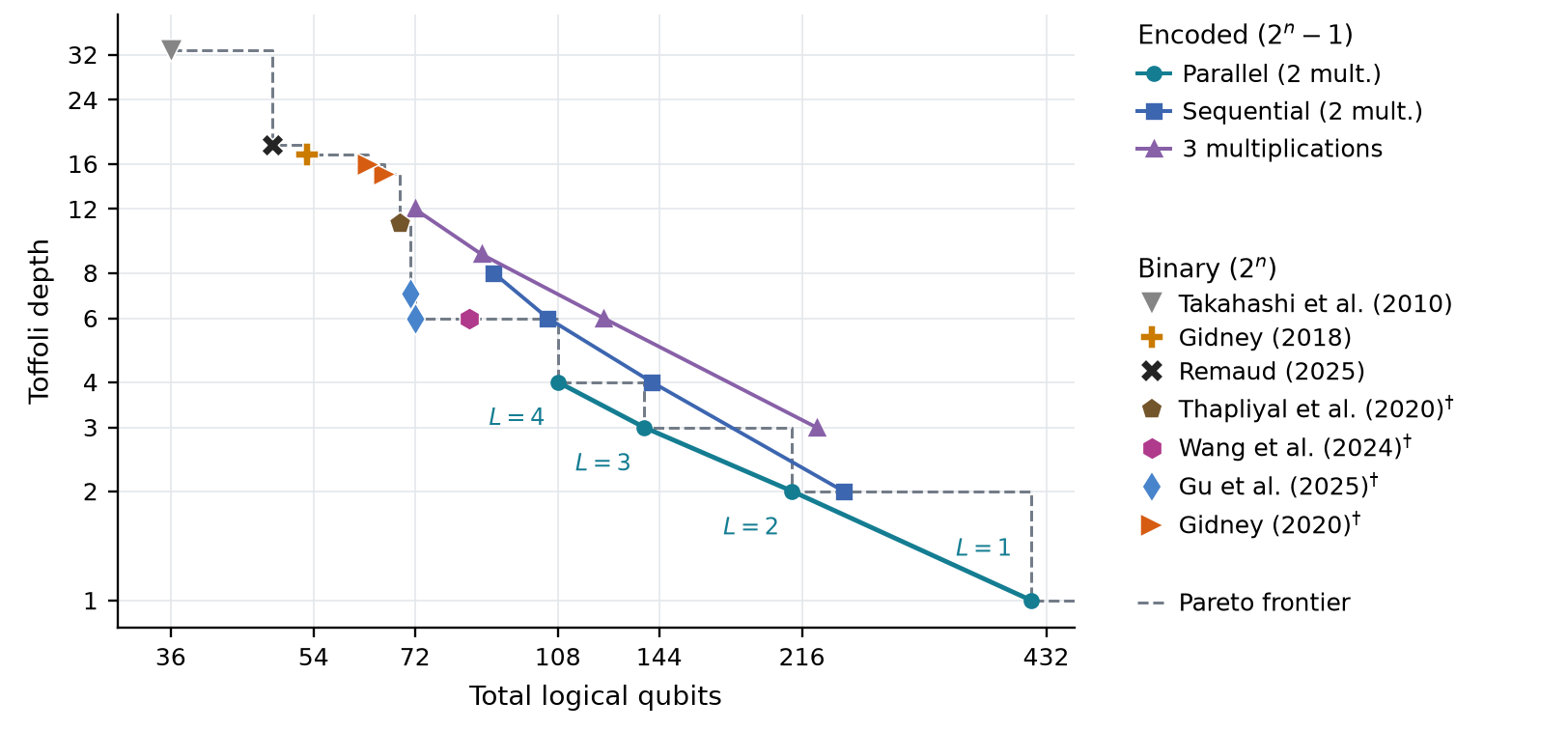}
\caption{Space--Toffoli-depth tradeoffs for $n=18$. Total qubit counts include all data and workspace registers. Both axes are logarithmic. The encoded schemes correspond to the three constructions in Table~\ref{tab:adder_resource_comparison} with $L=1,\ldots,4$. The dashed line marks the Pareto frontier of the evaluated binary
and encoded implementations. Daggers ($\dagger$) mark out-of-place adders converted to in-place
adders using Gidney's measurement-based cleanup~\cite{Gidney2018}.}
\label{fig:space_depth}
\end{figure}

Figure~\ref{fig:space_depth} shows the resulting comparison for $n=18$. The shallowest binary implementation considered has Toffoli depth $6$ and uses $72$ qubits. Below this depth, only our encoded constructions are available: the parallel construction ranges from Toffoli depth $4$ with $108$ qubits to Toffoli depth $1$ with $414$ qubits. Thus, below Toffoli depth $6$, the Pareto frontier consists of the parallel construction with $L=1,\ldots,4$.

\section{Application: programmable rotation}
\label{sec:applications}

In this section, we consider programmable rotations, which play an
important role in quantum computation~\cite{Gidney2018,
vonBurg2021Catalysis,Lee2021THC,Low2021multiplexed,Kim2025catalytic}.
A programmable rotation is the operation
\begin{equation}
\label{eq:programmable_rotation}
    \ket{a}\ket{\psi}
    \longmapsto
    \ket{a}e^{2\pi i aZ/N}\ket{\psi},
\end{equation}
where the angle is specified by an integer $a\in\mathbb Z_N$ stored in an $n$-qubit register, and $\ket{\psi}$ is the state of a single target qubit. In the standard binary representation, $N=2^n$; in our construction, $N=2^n-1$ and the register is replaced by one holding $\mathcal E_p(a)$.
To explain how our encoded addition can be used to implement such
rotations, let us first recall the relation between integer addition
and phase rotations.

In the standard binary representation, consider the phase-gradient state
\begin{equation}
    \ket{\Gamma_{2^n}}
    =
    \frac{1}{\sqrt{2^n}}
    \sum_{x=0}^{2^n-1}
    e^{-2\pi i x/2^n}\ket{x}.
\end{equation}
Adding $b$ to this register modulo $2^n$ multiplies it by the phase $e^{2\pi i b/2^n}$ and leaves it otherwise unchanged~\cite{Gidney2018}. Thus, the programmable rotation reduces to adding $\pm a$ to the phase-gradient register, with the sign controlled by the target qubit. Our encoding admits an analogous construction.

To perform phase kickback in our representation, define the encoded
phase-gradient state
\begin{equation}
\label{eq:encoded_phase_gradient}
    \ket{\Gamma_M}
    =
    \frac{1}{\sqrt M}
    \sum_{x=0}^{M-1}
    e^{-2\pi i x/M}\ket{\mathcal E_0(x)},
    \qquad M=2^n-1.
\end{equation}
Let $\widetilde T_b$ denote multiplication by $f^b$ on the encoded
register. Since $f^b f^x=f^{x+b}$, a change of summation variable
gives
\begin{equation}
    \widetilde T_b\ket{\Gamma_M}
    =
    e^{2\pi i b/M}\ket{\Gamma_M}.
\end{equation}
Thus field multiplication produces the required phase while returning
the encoded phase-gradient state unchanged. Throughout this section,
we assume that this ancillary state has already been prepared; see Ref.~\cite{Kim2025catalytic} for its preparation. Since it is returned unchanged, it needs to be prepared only once and can be reused throughout the computation. We note that without such a catalyst, i.e., starting from ancillas in $\ket{0}$, constant-depth synthesis of single-qubit rotations requires multi-qubit gates of growing width; with bounded-width gates, the depth must grow as $\Theta(\log\log(1/\varepsilon))$ in the worst case~\cite{Vasconcelos2026}.

To obtain the two opposite phases required by the target qubit, we
store the angle in the inverse-pair encoding
$\mathcal E_p(a)=(f^a,f^{-a})$. A swap of these two field registers,
controlled by the target, places
\begin{equation}
    u_q=f^{(-1)^q a}
\end{equation}
in the first register and $u_q^{-1}$ in the second. Multiplication by
the first register then produces the phase
$e^{2\pi i(-1)^q a/M}$. Reversing the controlled swap restores the
original ordering of the angle pair. The inverse rotation is obtained
by reversing which factor is selected.

It remains to explain how this multiplication is performed cleanly.
For a phase-gradient basis component $w=f^x$, first compute
$z=u_qw$ in a fresh register. The old register can then be erased
without further Toffoli gates: the identity
\begin{equation}
    w=u_q^{-1}z
\end{equation}
expresses its contents as a bilinear function of two registers that
remain, so Lemma~\ref{lem:bilinear_erasure} applies. Swapping the new
product into the cleared register completes the in-place
multiplication. After restoring the angle pair, the joint action on a computational basis state $\ket{q}$ of the target qubit is
\begin{equation}
\label{eq:encoded_programmable_rotation}
\begin{aligned}
    &\ket{\mathcal E_p(a)}\ket q\ket{\Gamma_M}\ket 0_W\\
    &\qquad\longmapsto
    e^{2\pi i(-1)^q a/M}
    \ket{\mathcal E_p(a)}\ket q\ket{\Gamma_M}\ket 0_W,
\end{aligned}
\end{equation}
where $W$ denotes the workspace. By linearity, this implements Eq.~\eqref{eq:programmable_rotation} with $N=M$. This establishes that both the
angle data and the phase-gradient state are preserved, and that the
workspace is returned to its initial state.

The controlled swaps introduce an additional cost beyond the field
multiplication. We fan out the target using CNOT gates, perform the
$n$ bitwise controlled swaps in parallel, and reverse the swaps and
fanout after the multiplication and erasure. Each controlled swap
is a single Toffoli gate up to CNOT gates, so each swap layer uses $n$
Toffoli gates and has Toffoli depth one. If the multiplication
has Toffoli depth $L$, the rotation therefore has Toffoli depth and Toffoli count
\begin{equation}
    D_{\mathrm{rot}}=L+2,
    \qquad
    N_{\mathrm{rot}}=\mu(n)+2n.
\end{equation}
The two additional layers select and restore the sign. The
space--depth tradeoffs of Section~\ref{sec:optimized} consequently
carry over to the programmable rotation: reducing the multiplication
depth requires additional workspace.

We now describe where the angle data arise and when they can be
reused. Electronic-structure algorithms based on double
factorization or tensor hypercontraction use Givens rotations to
change the orbital basis~\cite{vonBurg2021Catalysis,Lee2021THC}.
For a given rotation position, the angle depends on an index $j$
that may be in superposition. The corresponding operation is a
multiplexed rotation,
\begin{equation}
\label{eq:multiplexed_rotation}
    U_{\boldsymbol\theta}
    =
    \sum_{j=0}^{c-1}
    |j\rangle\!\langle j|
    \otimes e^{2\pi i\theta_j Z}.
\end{equation}
These rotations account for a substantial part of the non-Clifford
cost in the cited simulations~\cite{Low2021multiplexed}.

Choose integers $a_j\in\mathbb Z_M$ such that $a_j/M$ approximates
$\theta_j$ modulo one. Since the table is classical, its entries
$\mathcal E_p(a_j)$ can be computed before circuit execution and
loaded directly by QROM~\cite{Babbush2018}:
\begin{equation}
    |j\rangle|0\rangle^{\otimes 2n}
    \longmapsto
    |j\rangle|\mathcal E_p(a_j)\rangle.
\end{equation}
Applying Eq.~\eqref{eq:encoded_programmable_rotation} and reversing
the lookup implements the multiplexed rotation with angles
$a_j/M$. No coherent conversion from binary angles to field
encodings is needed.

Note that the loaded angle can be used more than once. If the same angle data are needed several times within a block and the index register is preserved, the encoded entry can remain loaded until the end of the block. (The uses may involve different Pauli rotations or inverse rotations.) For instance, the number of uses is $K=2$ in Ref.~\cite{vonBurg2021Catalysis} and $K=4$ in Ref.~\cite{Lee2021THC}, as counted in Ref.~\cite{Low2021multiplexed}.

In the standard binary representation, each rotation requires an addition whose Toffoli depth is $\Theta(n)$ for ripple-carry adders~\cite{Cuccaro2004,Gidney2018} and $\Theta(\log n)$ for carry-lookahead adders~\cite{Draper2006}. In contrast, our rotation has Toffoli depth $L+2$, independent of $n$. However, the cost of the lookup and the cleanup is common to both, and it is amortized only over the $K$ uses. As such, for the values of $K$ above, the overall improvement would be a modest constant factor.

\section{Discussion}
\label{sec:discussion}

In this paper, we showed that integer addition modulo $2^n-1$ can be implemented in constant Toffoli depth, provided that the integers are encoded as powers of a primitive element of $\GF{2^n}$. The key idea was to replace addition by multiplication in $\GF{2^n}$, which can be done in constant Toffoli depth using only $O(n)$ Toffoli gates~\cite{ChudnovskyChudnovsky1988,Jang2023,Liao2025}. By storing each number together with its inverse, the addition can also be done in place, and with measurement-based uncomputation the Toffoli depth can be brought down to one.

Our encoding is non-standard, and converting it to and from the standard binary representation is not cheap. However, there are settings in which this point is moot. Programmable rotations are one such example: the encoded angles can be loaded directly from a classical table, so no conversion is ever needed~\cite{Kim2025catalytic}. It will be interesting to find other applications of this kind.

There are several directions worth pursuing. First, the constants in our construction are still large, mainly due to $\mu(n)$. Any improvement in the bilinear complexity of multiplication in $\GF{2^n}$ would directly translate to a reduction in both the space and the Toffoli count. Second, it is natural to ask whether other moduli, such as $2^n$, admit a similar construction. We leave these questions for future work.

\section*{AI Disclosure}
The main idea of this work, an addition scheme with constant Toffoli depth using a linear number of ancillas, was conceived by the authors. Claude Opus 5.5 (Anthropic) was used to edit the manuscript, to write the code that generated the plots, and to optimize the constants of our constructions. All AI-generated content was reviewed and verified by the authors, who take full responsibility for the contents of this paper.

\section*{Acknowledgements}
We acknowledge support from the Advanced Scientific Computing Research program in the Office of Science of the Department of Energy (DE-SC0026109). This work was done in part while IK was visiting the Simons Institute for the Theory of Computing. M.L. acknowledges support from the National Research Foundation of Korea (NRF) through a grant funded by the Ministry of Science and ICT (Grant No. RS-2025-00515537), the Institute for Information \& Communications Technology Promotion (IITP) grants funded by the Korean government (MSIT) (Grants Nos. RS-2019-II190003 and RS-2025-02304540), and the Korea Institute of Science and Technology Information (Grant No. P26028).

\bibliographystyle{myhamsplain2}
\bibliography{ref}

\end{document}